\documentclass[pra,aps,twocolumn,superscriptaddress,longbibliography]{revtex4-1}
\usepackage[colorlinks=true, citecolor=red, urlcolor=blue ]{hyperref}
\usepackage{graphicx}
\usepackage{bm}
\usepackage{amsmath,amsfonts}
\usepackage{amsthm}
\usepackage{hyperref}
\usepackage{xcolor}
\hypersetup{
    urlcolor=magenta,
    citecolor=blue
           }

\usepackage{braket}
\theoremstyle{plain}
\usepackage{color}
\usepackage{amssymb}
\usepackage{amsthm}
\usepackage{amsfonts}
\usepackage{float}
\usepackage{tabularx}
\usepackage{graphicx}

\usepackage{mathtools}
\usepackage{esvect}
\usepackage{wrapfig}
\usepackage{amsthm}
\usepackage{verbatim}
\usepackage{bbm}
\usepackage[normalem]{ulem}

\usepackage{enumitem}
\usepackage{fmtcount}
\usepackage{booktabs}
\usepackage{csquotes}
\usepackage{epsfig}

\usepackage{tabularx}
\usepackage{graphicx}
\usepackage{amsmath}
\usepackage{braket}
\usepackage{latexsym}
\usepackage{bm}
\usepackage{graphics,epstopdf}
\usepackage{enumitem}
\usepackage{fmtcount}
\usepackage{booktabs}
\usepackage{csquotes}
\usepackage{epsfig}

\theoremstyle{plain}

\def\bea{\begin{eqnarray}}
\def\eea{\end{eqnarray}}
\def\ba{\begin{array}}
\def\ea{\end{array}}

\def\beq{\begin{equation}}
\def\eeq{\end{equation}}

\usepackage[normalem]{ulem}
\usepackage{float}
\usepackage{graphicx}  
\usepackage{dcolumn}          
\usepackage{amssymb}
\usepackage{appendix}
\usepackage{physics}   
\usepackage{mathtools}
\usepackage{esvect}
\usepackage{wrapfig}
\usepackage{amsthm}
\usepackage{verbatim}
\usepackage{bbm}

\usepackage[mathscr]{euscript}
\def\Tr{\operatorname{Tr}}

\def\({\left(}
\def\){\right)}
\def\[{\left[}
\def\]{\right]}

\newtheorem{theorem}{Theorem}

\newtheorem{corollary}{Corollary}

\newtheorem{lemma}{Lemma}

\begin{document}
\title{\textbf { Entanglement is indispensable for masking arbitrary set of quantum states}}
\author{Debarupa Saha}
\affiliation{Harish-Chandra Research Institute,  A CI of Homi Bhabha National Institute, Chhatnag Road, Jhunsi, Prayagraj  211019, India}
\author{Priya Ghosh}
\affiliation{Harish-Chandra Research Institute,  A CI of Homi Bhabha National Institute, Chhatnag Road, Jhunsi, Prayagraj  211019, India}
\author{Ujjwal Sen}
\affiliation{Harish-Chandra Research Institute,  A CI of Homi Bhabha National Institute, Chhatnag Road, Jhunsi, Prayagraj  211019, India}
\begin{abstract}
   We question the role of entanglement in masking quantum information contained in a set of mixed quantum states. We first show that a masker that can mask any two single-qubit pure states, can mask the entire set of mixed states comprising of the classical mixtures of those two pure qubit states as well. We then try to find the part played by entanglement in masking two different sets: One, a set of mixed states formed by the classical mixtures of two single-qubit pure commuting states, and another, a set of mixed states obtained by mixing two single-qubit pure non-commuting states. For both cases, we show that the masked states remain entangled unless the input state is an equal mixture of the two pure states. This in turn reveals that entanglement is necessary for masking an arbitrary set of two single qubit states, regardless of their mixednesses and mutual commutativity.
\end{abstract}
\maketitle
\section{Introduction}
Quantum communication has gained importance over its classical cousins in the past few years~\cite{r1}. One of the ruling reasons for this popularity is that quantum communication ensures confidentiality in secret-sharing processes, riding on several no-go theorems. 
No-go theorems in quantum information science impose certain restrictions or limitations on the occurrence of particular events in quantum  world, typically possible in the classical limit.
These no-go results 
include the no-cloning~\cite{cloning-dieks,cloning-yuen,r3,r4,cloning-simon,r5,r6}, no-broadcasting
~\cite{r7,r8,r9}, no-deletion~\cite{r10,r11}, no bit-commitment~\cite{r12,r13}, and no-masking theorems~\cite{r14,r15}. No-go theorems are important in various quantum information processing tasks, e.g., quantum key distribution~\cite{application-cloning-1}, quantum teleportation~\cite{r31}, and so on. 
The interest of this paper revolves around the notion of masking quantum information, encoded in quantum states.

The no-masking theorem, first introduced by Modi \textit{et al}.~\cite{r14}, 
shows that it is not possible to mask quantum information encoded in an arbitrary quantum state in a bipartite scenario.
Furthermore, it was shown in Ref.~\cite{r16} that the information contained in an arbitrary $k$-level 
quantum state can be masked into $m$-qudit multiparty system ($m>=4$) with the local dimension of the parties being smaller than $k$. It has been shown that masking of all $k$-level systems can be realized in a tripartite scenario with local dimensions $k$ or $k+1$~\cite{r17,r18}. Ref.~\cite{r19} proved the conjecture, given by Modi \textit{et al}.~\cite{r14}, which says that the maximal maskable states on the Bloch sphere with respect to any masker lie on a circle obtained by slicing a sphere with a plane. Masking of a restricted set of states was also addressed in several papers~\cite{r20,r21,r22}, and masking for non-hermitian quantum systems has been studied in~\cite{r23,r24,r25}. 
The relation between quantum maskers and error-correcting codes was explored in~\cite{r26}. Approximate to masking of quantum information has been addressed in Refs.~\cite{r27,r28}. 
Other interesting works in this realm include looking upon the processes where masking can be useful. For instance, Ref.~\cite{r29} showed that masking can have application in entanglement swapping-based quantum cryptography. Whereas, it was shown in Ref.~\cite{r30} that joint measurement in teleportation~\cite{r31} is actually a masking process. 
The effects of noise in the masking problem have also been discussed in detail~\cite{r33}. 
 Most of the masking problems so far considered the masking of pure states. However, there is a work~\cite{r32} considering masking of mixed states, where it has been shown that there exists an isometry that can mask a commuting subset of mixed states.

 For pure states, the only correlation that can mask the information is entanglement. 
 In this paper, we address the following question: Is masking of quantum information possible for mixed states without encoding the information in entanglement? We obtain
 that entanglement is a necessary correlation - resource - to mask an arbitrary set of two single-qubit states, irrespective of their mixednesses and mutual commutativity.

The paper is organized as follows:. In Sec.~\ref{sec-2}, we discuss the theory behind masking. In Sec.~\ref{sec-3}, we provide the analytics associated with masking of mixed states. 
Finally, we conclude in Sec. \ref{sec-4}.
\section{Quantum Masking}
\label{sec-2}
No-go theorems in quantum mechanics rule out the possibility of the occurrence of certain events in the quantum world.
The ``quantum masking," is one of the fascinating no-go theorems in quantum information science. In this section, we will briefly discuss ``quantum masking"~\cite{r14}.

Let us consider a maskable system $A$ with Hilbert space $\mathcal{H}_A$ and an auxiliary system $B$ with Hilbert space $\mathcal{H}_B$. 
Let $d_A$ and $d_B$ denote the dimensions of Hilbert spaces $\mathcal{H}_A$ and $\mathcal{H}_B$, respectively.
The composite Hilbert space of the maskable system and the auxiliary system will be $\mathcal{H}_{A}\otimes\mathcal{H}_{B}$. Let us assume that a set of quantum states of the maskable system, $\lbrace \zeta_{i} \rbrace_i$, and a state of the auxiliary system, $\ket{\varphi}$, are given.
One can construct a family of states $\lbrace \Omega_{i} \rbrace_i \coloneqq \lbrace \zeta_{i} \otimes \ket{\varphi}\bra{\varphi} \rbrace_{i}$ acting on the composite Hilbert space $\mathcal{H}_{A}\otimes\mathcal{H}_{B}$. Now, let us suppose that there exists an operation $\mathscr{S}$ that maps the family of states of a composite system, $\lbrace \Omega_{i} \rbrace_{i}$, to a set of states, $\lbrace \Lambda_i \rbrace_{i}$, acting on the same Hilbert space, such that 
the marginals of all states, $\lbrace \Lambda_i \rbrace_{i}$, corresponding to both the subsystems $A$ and $B$, are same, i.e.,
\begin{equation}
    \Tr_A [\Lambda_{i}]= \Tr_A[\Lambda_j] \hspace{4 mm}
\textnormal{and} \hspace{4 mm}
\Tr_B [\Lambda_{i}]= \Tr_B [\Lambda_j]
\end{equation}
holds for every $i \neq j$. Here, $\Tr_{A/B}$ denotes the partial trace over the first or second subsystem of the composite system. 
We will call the states $\lbrace \Lambda_i \rbrace_i$ as ``masked" states throughout the paper.
Hence, we cannot get any information about the states of the set $\lbrace \zeta_{i} \rbrace_i$, looking at the transformed family of states $\lbrace \Lambda_{i} \rbrace_i$ locally. In other words, the operator $\mathscr{S}$ has masked all the quantum information contained in the set $\lbrace\zeta_{i}\rbrace_{i}$ into the correlation between the maskable system and the auxiliary system. This process is called ``quantum masking," of the quantum information contained in the set of states, $\lbrace \zeta_{i} \rbrace_i$ and the operator $\mathscr{S}$ that made this process possible is called a quantum masker corresponding to the set.

Since entanglement is only the quantum correlation present in the pure states, if one masks the quantum information contained in a set consisting of pure states of the maskable system, then all the quantum information will be masked into the entanglement between the maskable system and auxiliary system. However, since any quantum correlation can be present in a mixed state, including entanglement, the quantum information contained in the set consisting of mixed states can be masked into any quantum correlation between the maskable system and auxiliary system. In the next section, we will show that entanglement is necessary in order to accomplish quantum masking of the set consisting of single-qubit mixed states.

Throughout the rest of the paper, we represent the maskable system and the auxiliary system by $A$ and $B$, respectively, and we will construct the composite system by the maskable system and auxiliary system. In the rest of the paper, we will consider the unitary operator as a quantum masker.

\section{Necessary criterion for masking of mixed states}
\label{sec-3}
In this section, we will discuss the role of entanglement in masking a set of single-qubit mixed states. We will primarily look for a masker of a set of at least any two mixed commuting states on $\mathbb{C}^2$ that masks the quantum information in different quantum correlations other than entanglement.
Then, we will try to find the same, but for a set consisting of at least any two mixed non-commuting states on $\mathbb{C}^2$.
Before going into our main results, we will discuss a little bit about the commuting and orthogonal sets of states.

\begin{itemize}
\item \textbf{Commuting set of states:-} A set of two states, $\rho_{1}$ and $\rho_{2}$, is called a commuting set if the states satisfy the condition $\rho_{1}\rho_{2}=\rho_{2}\rho_{1}$. Otherwise, the set will be called non-commuting.

\item \textbf{Orthogonal set of states:-} A set consisting of two states, $\rho_1$ and $\rho_2$, is said to be orthogonal if $\rho_1\rho_2 = 0$.
\end{itemize}

Any collection of orthogonal quantum states, whether pure or mixed, always satisfies the commuting condition. However, the converse holds for pure states only, i.e., the commuting set of pure states will always be an orthogonal set, but this is not the case for mixed states.


\begin{lemma}
\label{masking_mixed_from_pure}
Any masker that masks the quantum information contained in the set of two arbitrary pure states can simultaneously mask the entire set of mixed states formed by the convex mixtures of those two pure states. 
\end{lemma}

\begin{proof}
Let us assume that the maskable system, $A$, has a set of two pure states, $\ket{\psi_{1}}$ and $\ket{\psi_{2}}$, and the auxiliary system $B$ has a state, $\ket{\vartheta}$. The Hilbert spaces associated with the maskable system and auxiliary system can have arbitrary dimensions, and they need not be the same. The states of the composite system consisting of the maksable system and auxiliary system will be $(\ket{\psi_{1}} \otimes \ket{\vartheta})$ and $(\ket{\psi_{2}} \otimes \ket{\vartheta})$ and let us suppose that a masker $\mathscr{M}$ maps the states of the composite system into states $\ket{\Psi_{1}}$ and $\ket{\Psi_{2}}$, respectively, such that
\begin{align}
\label{lemma1-eqn-1}
\Tr_{A/B}[\ket{\Psi_{1}} \bra{\Psi_{1}}] = \Tr_{A/B}[\ket{\Psi_{2}} \bra{\Psi_{2}}] \equiv \rho_{A/B}.
\end{align}
Any mixed state along the line joining two pure states, $\ket{\psi_{1}}$ and $\ket{\psi_{2}}$, can be written as follows:
\begin{align}
\label{lemma1-eqn-2}
\Delta \coloneqq r \ket{\psi_{1}}\bra{\psi_{1}}+  (1-r)\ket{\psi_{2}}\bra{\psi_{2}},
\end{align}
with $r \geq 0$ and $0 \leq r \leq$ 1. Let us consider that the above masker $\mathscr{M}$ transforms $(\Delta \otimes \ket{\vartheta} \bra{\vartheta})$ into $\Tilde{\Delta}$. Then, we get
\begin{align}
\Tr_{A/B}\lbrack \Tilde{\Delta} \rbrack &= \Tr_{A/B}\left[ \mathscr{M} \left(\Delta \otimes
\ket{\vartheta} \bra{\vartheta} \right) \mathscr{M}^\dagger \right]  \label{lemma1-eqn-3}\\
&= r \Tr_{A/B}[\ket{\Psi_1}\bra{\Psi_1}] + (1-r) \Tr_{A/B}[\ket{\Psi_2}\bra{\Psi_2}] \label{lemma1-eqn-4},\\
&=\rho_{A/B} \label{lemma1-eqn-5}.
\end{align}
In Eq.~\eqref{lemma1-eqn-4}, we have substituted the form of $\Delta$ from Eq.~\eqref{lemma1-eqn-2} and used the fact that the masker is a linear operator. We have got the Eq.~\eqref{lemma1-eqn-5} using Eq.~\eqref{lemma1-eqn-1}. Eq.~\eqref{lemma1-eqn-5} holds true for all the mixed states formed by the pure states $\ket{\psi_{1}}$ and $\ket{\psi_{2}}$.

Hence, it completes the proof. 
\end{proof}




Conversely, we can state the following result:

\begin{corollary}
\label{sufficient-corollary}
Any masker corresponding to a set of any two arbitrary single-qubit mixed states will mask a set of two pure states that construct these mixed states.
\end{corollary}

\subsection{Masking of two single-qubit mixed commuting states}
\label{subsection-A}

Here, our main goal is 
to find a set of at least two single-qubit mixed commuting maskable states of the maskable system such that the states of the composite system transform into separable states under the action of the masker. 

\begin{lemma}
\label{pure-orthogonal-mask-lemma} 
The masked states must be maximally entangled to mask the quantum information contained in a set of two arbitrary single-qubit pure orthogonal states. 
\end{lemma}

\begin{proof}
Let us consider that the maskable system has a set consisting of two arbitrary single-qubit pure orthogonal states, $\ket{\Xi_{1}}$ and $\ket{\Xi_{2}}$ and
we want to mask the quantum information contained in the set.
Assume that the state of the auxiliary system is $\ket{a}$ acting on the two-dimensional Hilbert space and the masker corresponding to the set is $\mathscr{N}_1$.
The states of the composite system consisting of the maskable system and auxiliary system will be $(\ket{\Xi_{1}} \otimes \ket{a})$ and $(\ket{\Xi_{2}} \otimes \ket{a})$.
The masker $\mathscr{N}_1$ maps the states of the composite system to the states, $\ket{\chi_{1}}$ and $\ket{\chi_{2}}$, i.e.,
\begin{align}
\ket{\chi_{1}} \coloneqq \mathscr{N}_1 (\ket{\Xi_{1}} \otimes \ket{a}) \hspace{5 mm} \textnormal{and}, \hspace{5 mm} \ket{\chi_{2}} \coloneqq \mathscr{N}_1 (\ket{\Xi_{2}} \otimes \ket{a}).
\end{align}
%
Since the masker $\mathscr{N}_1$ is a global unitary operator, the application of $\mathscr{N}_1$ will preserve the orthogonality of any two pure orthogonal states. Hence, $\ket{\chi_{1}}$ and $\ket{\chi_{2}}$ will be orthogonal to each other.
Then, according to Ref.~\cite{PhysRevLett.85.4972}, we can express
the pure orthogonal states $\ket{\chi_{1}}$ and $\ket{\chi_{2}}$ on $\mathbb{C}^2 \otimes \mathbb{C}^2$ as follows:
\begin{align}
\ket{\chi_{1}} &= \sqrt{\alpha_{1}}\ket{\lambda_{0}\eta_{0}}+\sqrt{1-\alpha_{1}}\ket{\lambda_{1}\eta_{1}},  \label{walgate-pure-orthogonal-eqn-1}\\
\ket{\chi_{2}} &= \sqrt{\alpha_{2}}\ket{\lambda_{0}\eta_{0}^\perp}+\sqrt{1-\alpha_{2}}\ket{\lambda_{1}\eta_{1}^\perp}. \label{walgate-pure-orthogonal-eqn-2}
\end{align}
Here, $\alpha_1$ and $\alpha_2$ are normalized constants with $0 \leq \alpha_1, \alpha_2 \leq 1$.
$\ket{\eta_i^\perp}$ is the orthogonal state corresponding to the state $\ket{\eta_i}$
for all $i \in \lbrace 0,1 \rbrace$ and
$\lbrace \ket{\lambda_0}, \ket{\lambda_1} \rbrace$ forms an orthonormal basis.

To keep things simple, we will write $\lbrace \ket{\lambda_0}, \ket{\lambda_1} \rbrace$ as the standard computational basis.
Any single-qubit state in the Bloch representation can be written as
\begin{equation*}
    \ket{\Upsilon}=\cos{\frac{x}{2}}\ket{0}+ e^{iy}\sin{\frac{x}{2}}\ket{1},
\end{equation*}
with $0\leq x \leq \pi$ and $0 \leq y \leq 2\pi$. 
Without loss of generality, we can consider $\ket{\eta_{0}}$ and $\ket{\eta_{1}}$ in the following form:
 \begin{align}
\ket{\eta_{0}} &= \cos{\frac{\theta'}{2}}\ket{0}+\sin{\frac{\theta'}{2}}\ket{1}, \\
\ket{\eta_{1}} &= \cos{\frac{\theta}{2}}\ket{0}+\sin{\frac{\theta}{2}}\ket{1},
 \end{align}
with $\theta',\theta \in \lbrack 0, 2 \pi \rbrack$.
However, the proof will remain same if one considers 
any arbitrary values of local phases in the forms of the states $\ket{\eta_{0}}$ and $\ket{\eta_{1}}$.
Now, substituting the mathematical forms of $\ket{\eta_{0}}$, $\ket{\eta_{1}}$ into Eqs.~\eqref{walgate-pure-orthogonal-eqn-1} and \eqref{walgate-pure-orthogonal-eqn-2}, we get
 \begin{align}
\ket{\chi_{1}} &=\sqrt{\alpha_{1}}(\cos{\frac{\theta'}{2}}\ket{00}+\sin{\frac{\theta'}{2}}\ket{01}) \nonumber \\ &+\sqrt{1-\alpha_{1}}(\cos{\frac{\theta}{2}}\ket{10}+\sin{\frac{\theta}{2}}\ket{11}), \label{eq-chi1}\\
\ket{\chi_{2}} &= \sqrt{\alpha_{2}}(\sin{\frac{\theta'}{2}}\ket{00}-\cos{\frac{\theta'}{2}}\ket{01}) \nonumber \\ &+ \sqrt{1-\alpha_{2}}(\sin{\frac{\theta}{2}}\ket{10}-\cos{\frac{\theta}{2}}\ket{11}),\label{eq-chi2}
\end{align}
 with $0 \leq \alpha_1, \alpha_2 \leq 1$.

It can be easily shown that the quantum masking criteria $\Tr_{A/B} [\ket{\chi_{1}}\bra{\chi_{1}}] =\Tr_{A/B} [\ket{\chi_{2}}\bra{\chi_{2}}]$ will be satisfied by $\ket{\chi_{1}}$ and $\ket{\chi_{2}}$ when $\sqrt{\alpha_{1}}=\pm\frac{1}{\sqrt{2}},\sqrt{\alpha_{2}}=\pm\frac{1}{\sqrt{2}}$, and $\theta' = \theta + \pi$, for any arbitrary $\theta'$ and $\theta$. 
Hence, $\ket{\chi_{1}}$ and $\ket{\chi_{2}}$ will be the Bell states upto local unitaries when they satisfy the aforementioned requirements. 
Thus, the proof is completed.


\end{proof}

\begin{theorem}
\label{neccessary-sufficient-commuting}
Entanglement in masked states is necessary to mask the set of two arbitrary single-qubit mixed commuting states.
\end{theorem}

\begin{proof}
If two arbitrary single-qubit mixed states are commuting, then both these mixed states will be just different convex mixtures of two single-qubit pure orthogonal states, let's say $\ket{\Xi_{3}}$ and $\ket{\Xi_{4}}$. Let us consider a set of two arbitrary single-qubit mixed commuting states in the maskable system, denoted as $\lbrace \Gamma_i \rbrace_i$ with $i \in \lbrace 1,2 \rbrace$, and a state in the auxiliary system, $\ket{b}$, acting on two-dimensional Hilbert space. Each mixed states in the set of the maskable system can be written as 
\begin{align}
\Gamma_i \coloneqq p_i \ket{\Xi_{3}}\bra{\Xi_{3}} +  (1-p_i)\ket{\Xi_{4}}\bra{\Xi_{4}}, \label{theorem-1-eqn-11}
\end{align}
with $i \in \lbrace 1,2 \rbrace$  and $0 \leq \lbrace p_i \rbrace_{i=1, 2} \leq 1$. 

Assume that the quantum masker corresponding to the set of mixed states, $\lbrace \Gamma_i \rbrace_i$ with $i \in \lbrace 1, 2 \rbrace$, is $\mathscr{N}_2$ which is a global unitary operator.
From Corollary \ref{sufficient-corollary}, we can say that 
the quantum masker $\mathscr{N}_2$
can simultaneously mask the set of pure states, 
$\ket{\Xi_{3}}$ and $\ket{\Xi_{4}}$.
The masker $\mathscr{N}_2$ maps the states of the composite system, 
$\lbrace \Gamma_i \otimes \ket{b}\bra{b} \rbrace_i$ to states $\lbrace \rho_i \rbrace_i$, i.e.,
\begin{align}
   \rho_i &\coloneqq \mathscr{N}_2 \left(\Gamma_i \otimes \ket{b}\bra{b} \right) \mathscr{N}_2^\dagger,\\
    &= p_i \ket{\chi_{3}} \bra{\chi_{3}} + (1- p_i) \ket{\chi_{4}} \bra{\chi_{4}} \label{theorem-1-eqn-15},
\end{align}
with $0 \leq p_i \leq 1$ where $i \in \lbrace 1, 2 \rbrace$. The mathematical expressions of states $\ket{\chi_{3}}$ and $\ket{\chi_{4}}$ will be the same as the states $\ket{\chi_{1}}$ and $\ket{\chi_{2}}$, described in Eqs.~\eqref{eq-chi1} and~\eqref{eq-chi2}, respectively 
which satisfy the condition $\sqrt{\alpha_{1}}=\pm\frac{1}{\sqrt{2}},\sqrt{\alpha_{2}}=\pm\frac{1}{\sqrt{2}}$, and $\theta' = \theta + \pi$ for all $\theta \in \lbrack 0, \pi \rbrack$. We obtain the Eq. \eqref{theorem-1-eqn-15}, substituting the mathematical forms of $\lbrace \Gamma_i \rbrace_i$ from Eq. \eqref{theorem-1-eqn-11} and using Lemma \ref{pure-orthogonal-mask-lemma}.


 The criterion by Horodecki \textit{et al.}~\cite{ent} states that a bipartite state $\rho_{AB}$ on the Hilbert space of arbitrary dimension is definitely entangled if the von-Neumann entropy of the reduced subsystem ($\rho_{A/B} \coloneqq \Tr_{A/B} [\rho]$) of corresponding state is greater than the von-Neumann entropy of total state, i.e.,
$S(\rho_{A/B}) > S(\rho_{AB})$. The von-Neumann entropy of any state $\rho_{AB}$ is defined as $S(\rho_{AB}) \coloneqq - \Tr \left[ \rho_{AB} \log_2{\rho_{AB}} \right]$.
However, if any bipartite state does not follow the above condition, we cannot conclude anything about the separability of the state (i.e., it can be either separable or entangled).

Let us denote the reduced subsystems of the states $\lbrace \rho_i \rbrace_i$ by $\lbrace \rho^i_{A/B} \rbrace_i$ where $i \in \lbrace 1, 2 \rbrace$.
In our case, the von-Neumann entropies ($S(\rho^i_{A/B})$) of reduced subsystem of the states $\lbrace \rho_i \rbrace_i$
are unity, 
whereas the von-Neumann entropies of both states, $\lbrace \rho_i \rbrace_i$, are 
\begin{align*}
   &S(\rho_i) = -p_i\log_{2}[p_i] - \\&\frac{1}{2} (1  + |1 - 2 p_i|)\log_{2} \Bigg[\frac{(1 + |1 - 2 p_i|)}{2}\Bigg],
\end{align*}
for $0 \leq p_i \leq \frac{1}{2}$, and

\begin{align*}
   S(\rho_i) &= -(1-p_i)\log_{2}[1-p_i]\\&- 
\frac{1}{2}(1 + |1 - 2 p_i|)\Bigg[\frac{(1 + |1 - 2 p_i|)}{2}\Bigg],
\end{align*}
for $\frac{1}{2} < p_i \le 1$ where $i \in \lbrace 1, 2 \rbrace$.
We can observe that $S(\rho_i)$ for both the states $\lbrace \rho_i \rbrace_i$ are solely dependent on $p_i$ and not on $\theta$ with $i \in \lbrace 1, 2 \rbrace$. 
We plot the von-Neumaan entropies of both states $(S(\rho_i))$ and the von-Neumaan entropies of the reduced subsystem of both states $(S(\rho^i_{A/B}))$ with respect to $p_i$ in Fig.~\ref{fig:1}.
It shows that $S(\rho_i) \leq S(\rho^i_{A/B})$ for all $0 \leq p_i \leq 1$ with $p_i = \frac{1}{2}$ being the point of equality where $i \in \lbrace 1, 2 \rbrace$. Therefore, according to Ref.~\cite{ent}, the mixed states, $\lbrace \rho_i \rbrace_i$,
will be entangled states for all $ p_i \in \left[0, \frac{1}{2}) \cup ( \frac{1}{2}, 1 \right]$ where $i \in \lbrace 1, 2 \rbrace$. Since the Hilbert space of the composite system is $\mathbb{C}^2 \otimes \mathbb{C}^2$, we can apply PPT criterion~\cite{3,PPT} to detect entanglement in the states, $\lbrace \rho_i \rbrace_i$, at $p_i = \frac{1}{2}$ for all $i \in \lbrace 1, 2 \rbrace$.
Applying PPT criterion, it can be proved that both the mixed states $\lbrace \rho_i \rbrace_i$, for all $i \in \lbrace 1, 2 \rbrace$, are separable at $p_i = \frac{1}{2}$. 
If both the masked states are separable then the only possibility is that these states are identical with both having same $p$ values i.e. $p_1=p_2=0.5$. 
The theory of quantum masking demands that the set of states to be masked contains at least two states, our set of masked states
should contain at least one entangled state. This proves that entanglement is an absolute necessity for masking the information contained in two arbitrary single-qubit mixed commuting states.

\begin{figure}
\centering
\includegraphics[width=\linewidth]{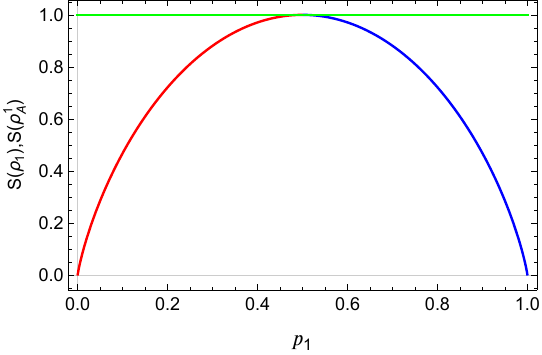}
\caption{\textbf{\textit{Behavior of global and local entropies of masked state, $\rho_1$, with respect to the mixedness $p_1$.}}
The global entropy $S(\rho_1)$ and local entropy $S(\rho^1_{A})$ of the masked state, $\rho_1$, are plotted along the vertical axis, whereas the mixedness, $p_1$, is plotted along the horizontal axis.
Here, the green line represents the entropy of the local reduced density matrix of the masked state $\rho_1$, while the red line shows the global entropy of the masked state $\rho_1$ for $p_1 \leq \frac{1}{2}$ and the blue line for $p > 
\frac{1}{2}$. It implies that the masked state $\rho_1$ is entangled for all $p_1$ except at $p_1 = \frac{1}{2}$. The behaviors of $S(\rho_2)$ and $S(\rho^2_{A})$ for masked state $\rho_2$ with mixedness $p_2$ will be exactly the same as those for $\rho_1$. The vertical axis is in units of bits, and the horizontal axis is dimensionless.}
 \label{fig:1}
 \end{figure}

\end{proof}

The above Theorem~\ref{neccessary-sufficient-commuting} can be generalized to a set of more than two single-qubit mixed commuting states, using the similar argument.
In Appendix~\ref{Ex1}, we present a numerical example to support our findings related to the masking of single-qubit mixed commuting states.

\subsection{Masking of two single-qubit mixed non-commuting states}
Let us now move into another main result of our work, which is that entanglement is a necessary ingredient in the masking of a set consisting of at least two single-qubit mixed non-commuting states. 
First, we will start with a statement on the masking of two arbitrary single-qubit pure non-orthogonal states, which states that

\begin{lemma}
\label{l3}
    Masking of two arbitrary single-qubit pure non-orthogonal states requires the masked states to be equally entangled.
\end{lemma}

\begin{proof}
Let us suppose that two arbitrary pure non-orthogonal states of the maskable system, $\ket{\Phi_1}$ and $\ket{\Phi_2}$ on $\mathbb{C}^2$, and a state of the auxiliary maskable system, $\ket{c}$, on $\mathbb{C}^2$, are given. 
The states of the composite system will be $(\ket{\Phi_1} \otimes  \ket{c})$ and $(\ket{\Phi_2} \otimes \ket{c})$.

According to Ref.~\cite{ortho,nonortho}, one can write any two arbitrary pure non-orthogonal states, $\ket{\phi_1}$ and $\ket{\phi_2}$, on $\mathbb{C}^d \otimes \mathbb{C}^d$ as follows:
\begin{align*}
 \ket{\phi_{1}} &\coloneqq \sum_{k=0}^{d-1} \sqrt{r_k} \ket{\gamma_k \tau_{k}},\\
\ket{\phi_{2}} &\coloneqq \sum_{k=0}^{d-1} \sqrt{s_k}\ket{\gamma_k \nu_{k}},
\end{align*}
where $\lbrace \ket{\gamma_k} \rbrace_k$ forms an orthonormal basis set on the first subsystem of the composite system, and $0 \leq r_k, s_k \leq 1$. $\lbrace \tau_k \rbrace_k$ and $\lbrace \nu_k \rbrace_k$ are normalized, non-orthogonal and
they satisfy the condition: $\sqrt{r_{k}s_{k}}\braket{\tau_{k}} {\nu_{k}} = \sqrt{r_{l}s_{l}}\braket{\tau_{l}} {\nu_{l}}$ for all $k,l \in \lbrace 0,1, \ldots , (d-1) \rbrace$.
If $\braket{\tau_{k}} {\nu_{k}} = 0$ for all $k$, then $\ket{\phi_1}$ and $\ket{\phi_2}$ will be orthogonal with each other. 

In our scenario, assume that the masker maps the states of the composite system, $(\ket{\Phi_1} \otimes  \ket{c})$ and $(\ket{\Phi_2} \otimes \ket{c})$, into two pure non-orthogonal states, $\ket{\Sigma_1}$ and $\ket{\Sigma_2}$ on $\mathbb{C}^2 \otimes \mathbb{C}^2$, respectively.
Using the results of Ref.~\cite{ortho,nonortho}, $\ket{\Sigma_1}$ and $\ket{\Sigma_2}$ on $\mathbb{C}^2 \otimes \mathbb{C}^2$ can be expressed as follows:
\begin{align*}
 \ket{\Sigma_{1}} &\coloneqq \sqrt{t_{0}}\ket{0\tau_{0}}+\sqrt{1-t_{0}}\ket{1\tau_{1}},\\
 \ket{\Sigma_{2}} &\coloneqq \sqrt{t_{1}}\ket{0\nu_{0}}+\sqrt{1-t_{1}}\ket{1\nu_{1}},
\end{align*}
with $\sqrt{t_0t_1}\braket{\tau_{0}} {\nu_{0}} = \sqrt{(1-t_0)(1-t_1)}\braket{\tau_{1}} {\nu_{1}}$ and $0 \leq t_0, t_1 \leq 1$. $\lbrace \ket{0}, \ket{1} \rbrace$ represents the standard computational basis.
Here, without loss of generality, we can consider 
the following forms of $\ket{\tau_{0}}, \ket{\tau_{1}}, \ket{\nu_{0}}$, and $\ket{\nu_{1}}$
\begin{align*}
\ket{\tau_{0}}&\coloneqq \ket{0}, \\
\ket{\tau_{1}} &\coloneqq \cos{\Theta_{1}}\ket{0}+\sin{\Theta_{1}}\ket{1},\\
\ket{\nu_{0}} &\coloneqq \cos{\Theta_{2}}\ket{0}+e^{i\varphi_{1}}\sin{\Theta_{2}}\ket{1}, \hspace{5 mm} \\
\ket{\nu_{1}} &\coloneqq \cos{\Theta_{3}}\ket{0}+e^{i\varphi_{2}}\sin{\Theta_{3}}\ket{1},
\end{align*}
with $\Theta_{1,2,3} \in [0,\pi]$ and $\varphi_{1,2} \in [0,2\pi]$. 

 It can be shown that the masking condition, $\Tr_{A/B}[\ket{\Sigma_{1}}\bra{\Sigma_{1}}) = \Tr_{A/B} [\ket{\Sigma_{2}}\bra{\Sigma_{2}}]$ can only be fulfilled if any of the following conditions are satisfied:

 \begin{itemize}
    \item \textbf{Condition 1:} $t_0=t_1= \frac{1}{2}, \varphi_1=\varphi2=0, \Theta_1=\frac{\pi}{2}$, and $\Theta_2=\Theta_3 - \frac{\pi}{2}$.\\
 The masked states under such condition will be,
 \begin{align}
     \ket{\Sigma^{1}_1} &\coloneqq \pm (1/\sqrt{2})(\ket{00}+\ket{11}, \label{condition-1-1}\\
     \ket{\Sigma^{1}_2} &\coloneqq \pm (1/\sqrt{2})(\ket{0\boldsymbol{\nu_0}}+\ket{1\boldsymbol{\nu^\perp_0}} \label{condition-1-2},
 \end{align}
where $\ket{\boldsymbol{\nu_0}} \coloneqq \cos{\boldsymbol{\Theta}}\ket{0}+ \sin{\boldsymbol{\Theta}}\ket{1}$ with $\boldsymbol{\Theta} \in [0, \pi]$. The state $\ket{\boldsymbol{\nu_0}^\perp}$ is orthogonal to the state $\ket{\boldsymbol{\nu_0}}$.
At $\boldsymbol{\Theta} = 0, \pi$, both the masked states are essentially the same. Hence, $\boldsymbol{\Theta} = 0,\pi$ is not being considered for masking. Since at $\boldsymbol{\Theta} = \frac{\pi}{2}$ the masked states will be orthogonal with each other, $\boldsymbol{\Theta} = \pi/2$ should also not be considered. So, we would like to note here that in such a scenario, $\boldsymbol{\Theta}$ can only be in the range $\boldsymbol{\Theta} \in (0, \frac{\pi}{2}) \cup (\frac{\pi}{2}, \pi)$. 

\item \textbf{Condition 2:-} $t_0 = t_1 = \frac{1}{2}, \varphi_1 = 0, \Theta_2 = \Theta_1,
\Theta_3 = 0$, and $\varphi_2 \in [0,2\pi]$. \\
Under this condition, the masked states will be
\begin{align}
\ket{\Sigma^{2}_1} &\coloneqq \pm (1/\sqrt{2})(\ket{00}+\ket{1\boldsymbol{\tau_1}}, \label{condition-2-1}\\
\ket{\Sigma^{2}_2} &=\pm (1/\sqrt{2})(\ket{0\boldsymbol{\tau_1}}+\ket{10}, \label{condition-2-2}
 \end{align}
where $\ket{\boldsymbol{\tau_1}} \coloneqq \cos{\boldsymbol{\Theta}}\ket{0}+\sin{\boldsymbol{\Theta}}\ket{1}$ with $\boldsymbol{\Theta} \in [0, \pi]$. For the same reason as \textbf{Condition 1}, here the range of $\boldsymbol{\Theta}$ will be $\boldsymbol{\Theta} \in (0, \frac{\pi}{2}) \cup (\frac{\pi}{2}, \pi)$.

\end{itemize}


Since $\ket{\Sigma^{1,2}_1}$ and $\ket{\Sigma^{1,2}_2}$ are two-qubit pure states, we can use the local von Neumann entropy as a measure of entanglement to measure the entanglement of these masked states. 
Clearly, $S(\Tr_{A/B}[ \ket{\Sigma^{1}_1} \bra{\Sigma^{1}_1}]) = S(\Tr_{A/B}[ \ket{\Sigma^{1}_2} \bra{\Sigma^{1}_2}])= 1 $ and $S(\Tr_{A/B}[ \ket{\Sigma^{2}_1} \bra{\Sigma^{2}_1}]) = S(\Tr_{A/B}[ \ket{\Sigma^{2}_2} \bra{\Sigma^{2}_2}]) < 1$
for all $\boldsymbol{\Theta} \in (0, \frac{\pi}{2}) \cup (\frac{\pi}{2}, \pi)$. I.e., the entanglement content of the states, $\ket{\Sigma^{1,2}_1}$ and $\ket{\Sigma^{1,2}_2}$, are non-zero and equal when $\boldsymbol{\Theta} \in (0, \frac{\pi}{2}) \cup (\frac{\pi}{2}, \pi)$.
Hence, the masked states obtained from $\textbf{Condition 1}$ are both maximally entangled, and the masked states obtained from $\textbf{Condition 2}$ are equally entangled.
This concludes the proof. 
\end{proof}
This is the ideal moment to investigate the role of entanglement in the masking of two arbitrary single-qubit mixed non-commuting states.

\begin{theorem}
\label{theorem-non-commuting}
Entanglement in masked states is necessary to mask two arbitrary single-qubit mixed non-commuting states.
\end{theorem}

 \begin{figure}
     \centering
\includegraphics[width=\linewidth]{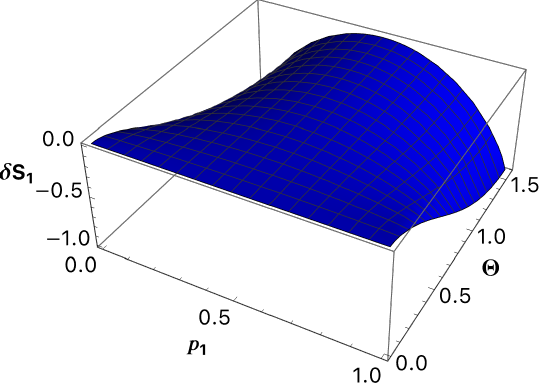}
     \caption{
     \textbf{\textit{Behavior of difference between global and local entropies $(\boldsymbol{\boldsymbol{\delta S_1}})$ of masked state, $\kappa_1$, against $\boldsymbol{p_1}$ and $\boldsymbol{\Theta}$.}}
     The discrepancy between global and local entropies of state $\kappa_1$, $\boldsymbol{\delta S_1}\coloneqq S(\kappa_{1})-S(\Tr_{A}[\kappa_{1}])$, is plotted along the vertical axis, and mixing parameters $\boldsymbol{p_1}$ and $\boldsymbol{\Theta}$ along the horizontal axes. $\boldsymbol{\boldsymbol{\delta S_1}}$ is symmetric about $\boldsymbol{\Theta} = \frac{\pi}{2}$. $\boldsymbol{\delta S_2}$ has exactly the same quantitative nature as $\boldsymbol{\delta S_1}$ against $\boldsymbol{p_2}$ and $\boldsymbol{\Theta}$.
     The vertical axis is in bits, while the horizontal axes is dimensionless (for $\boldsymbol{p_1}$) and is in unit of radian (for $\boldsymbol{\Theta}$).}

     \label{fig:2}
 \end{figure}

\begin{proof}
Any arbitrary two single-qubit mixed non-commuting states can be realized by different convex mixtures of the same two single-qubit pure non-orthogonal states. Let us consider that the two single-qubit pure non-orthogonal states are $\ket{\boldsymbol{\Sigma_1}}$ and $\ket{\boldsymbol{\Sigma_2}}$. The convex mixtures of these two pure states can be represented as
\begin{align*}
\boldsymbol{\kappa} &\coloneqq \boldsymbol{p} \ket{\boldsymbol{\Sigma_1}} \bra{\boldsymbol{\Sigma_1}} + (1-\boldsymbol{p}) \ket{\boldsymbol{\Sigma_2}} \bra{\boldsymbol{\Sigma_2}},
\end{align*}
with $0\leq \boldsymbol{p} \leq 1$. Here, our goal is to mask a set of two single-qubit mixed non-commuting states of the maskable system, let's say $\boldsymbol{\kappa_1}$ and $\boldsymbol{\kappa_2}$, parameterized by $\boldsymbol{p_1}$ and $\boldsymbol{p_2}$, respectively. An arbitrary single-qubit pure state in the auxiliary system, $\ket{d}$, is given. 

In the similar way as Theorem~\ref{neccessary-sufficient-commuting}, from Corollary \ref{sufficient-corollary} and Lemma $\ref{l3}$, we can say that if any masker will mask the set of states $\lbrace \boldsymbol{\kappa}_q \rbrace_q$ with $q \in \lbrace 1, 2 \rbrace$ considering the state $\ket{d}$ as an auxiliary state, then the masked states will be
\begin{align*}
    \kappa_q &\coloneqq \boldsymbol{p_q} \ket{\Sigma^g_q} \bra{\Sigma^{g}_q} + (1-\boldsymbol{p_q}) \ket{\Sigma^{g}_q} \bra{\Sigma^g_q},
\end{align*}
with $0 \leq \boldsymbol{p_q} \leq 1$ and $q,g \in \lbrace 1,2 \rbrace$. The mathematical expressions of $\ket{\Sigma^g_1}$ and $\ket{\Sigma^g_2}$ with $g \in \lbrace 1,2 \rbrace$ are described in Eqs.~\eqref{condition-1-1},~\eqref{condition-1-2},~\eqref{condition-2-1}, and~\eqref{condition-2-2}.

Note that when $g=1$, both the masked states will be entangled except at a single point, i.e., $\boldsymbol{p_{1,2}} = \frac{1}{2}$, and at $\boldsymbol{p_{1,2}} = \frac{1}{2}$, the states will be separable using the same reason described in Theorem~\ref{neccessary-sufficient-commuting}. 
Thus, entanglement in masked states is necessary for masking the information encoded in a set of two single-qubit mixed non-commuting states when $g = 1$.

Now, we are left with the case only where $g=2$.
In this case, we will find $\boldsymbol{\delta S_q} \coloneqq S(\kappa_{q})-S(\Tr_{A/B}[\kappa_{q}])$ at each $\boldsymbol{p_q}$ and $\boldsymbol{\Theta}$ for $q \in \lbrace 1,2 \rbrace$.  
To do so, we plot $\boldsymbol{\delta S_{1,2}}$ against $\boldsymbol{p_{1,2}}$ and $\boldsymbol{\Theta}$ in Fig.~\ref{fig:2}.  Notice in Fig.~\ref{fig:2} that $\boldsymbol{\delta S_{1,2}} < 0$ for $ \boldsymbol{p_{1,2}} \in [0, \frac{1}{2}) \cup ( \frac{1}{2}, 1]$ and $\boldsymbol{\Theta} \in (0, \frac{\pi}{2}) \cup (\frac{\pi}{2}, \pi)$.
However, at $\boldsymbol{p_{1,2}} = \frac{1}{2}$, $\boldsymbol{\delta S_{1,2}} = 0$ for every value of $\boldsymbol{\Theta} \in (0, \frac{\pi}{2}) \cup (\frac{\pi}{2}, \pi)$.
Hence, from the Horodecki \textit{et. al.} criteria~\cite{3} we can conclude that the masked states $\kappa_{1,2}$ are entangled for all $\boldsymbol{p_{1,2}} \in [0, \frac{1}{2}) \cup ( \frac{1}{2}, 1] $ and $\boldsymbol{\Theta} \in (0, \frac{\pi}{2}) \cup (\frac{\pi}{2}, \pi)$. 
We can find the entanglement content of the masked states $\kappa_{1,2}$ at $\boldsymbol{p_{1,2}} = \frac{1}{2}$ when $g=2$ as follows. At $\boldsymbol{p_{1,2}} = \frac{1}{2}$, determinants of the partially transposed matrix of $\kappa_{1,2}$, $\det(\kappa_{1,2}^{T_B})$, are zero for all values of $\boldsymbol{\Theta} \in (0, \frac{\pi}{2}) \cup (\frac{\pi}{2}, \pi)$, where $T_B$ denotes partial transpose over the second subsystem of the composite system. Hence, the masked states $\kappa_{1,2}$ are separable at $\boldsymbol{p_{1,2}} = \frac{1}{2}$ for all values of $\boldsymbol{\Theta} \in (0, \frac{\pi}{2}) \cup (\frac{\pi}{2}, \pi)$ as it satisfies the condition of separability for mixed states~\cite{s2,s1}.
Therefore, the masked states will be entangled for all $\boldsymbol{p_{1,2}}$ except $\boldsymbol{p_{1,2}} = \frac{1}{2}$ and at $\boldsymbol{p_{1,2}} = \frac{1}{2}$, the masked states are separable when $g=2$ is satisfied. 
This establishes that entanglement in masked states is necessary for masking information contained in a set of any two arbitrary single-qubit mixed non-commuting states when $g = 2$.

It ends the proof.
\end{proof}

The aforementioned Theorem~\ref{theorem-non-commuting} can be extended to a larger set of states, meaning the necessary criterion also applies when considering more than two single-qubit mixed non-commuting states.
Moreover, we provide a numerical example to support Theorem~\ref{theorem-non-commuting} in Appendix~\ref{Ex2}.

\section{conclusion}
\label{sec-4}

Masking of quantum information is a phenomenon by which information contained in a set of states is distributed in the correlation between the states and auxiliary states. The no-masking theorem restricts the masking of arbitrary qudit states in a bipartite scenario. It is always possible to mask a set of states situated on a circle of the Bloch sphere sliced by a plane with a unique masker. In this work, we considered two scenarios: One, masking a set consisting of two mixed states
formed by the linear combination of two single-qubit pure commuting states, and another, masking of a set consisting of two mixed states, but they are classical mixtures of two single-qubit pure non-commuting states. The aim was to find the role of entanglement in the masking of the states in both cases. Thus, we mainly addressed the following question in the paper:
``Is masking of mixed states possible if the masked states are separable?".

Entanglement is a fundamental requisite for the masking of mixed states, irrespective of whether they are commuting or not. In both scenarios, we found that the mixed state lying at the middle of the line joining the pure masked states is the only separable state in the set of masked states, while all else are entangled, for any masker. As masking is defined for a set consisting of at least two states, we can conclude that entanglement is necessary for 
masking the information of an arbitrary set of states. 

\section*{Acknowledgment}
We acknowledge partial support from the Department of Science and Technology, Government of India through the QuEST grant (grant number DST/ICPS/QUST/Theme3/2019/120).

\bibliography{Maskref}



\section*{Appendix}

\begin{appendix}

\subsection{Numerical verification of Theorem \ref{neccessary-sufficient-commuting}}
\label{Ex1}
Here we will provide a numerical verification of Theorem \ref{neccessary-sufficient-commuting}. We try to find the existence of separable masked states in masking a set of mixed states formed by the convex mixtures of two single-qubit pure orthogonal states, $\epsilon_{1} \coloneqq \ket{0}$ and $\epsilon_{2} \coloneqq \ket{1}$, where $\ket{0}$ and $\ket{1}$ denote the eigenstates of $\sigma_z$ with eigenvalues $1$ and $-1$, respectively.

\begin{figure}
     \centering
\includegraphics[width=\linewidth]{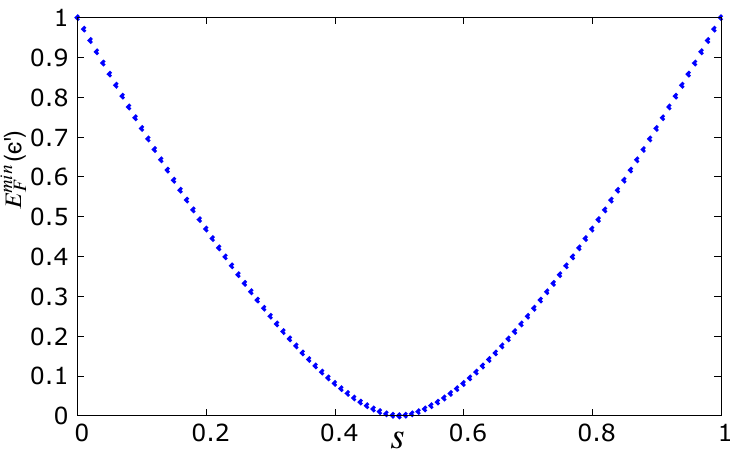}
\caption{\textbf{\textit{Behavior of entanglement ($E_F^{\textnormal{min}})$ of masked state ($\epsilon'$), corresponding to the optimized masker that minimizes the entanglement of formation of the set $\{\epsilon'\}$ with respect to the mixing parameter $s$}}. We use the entanglement of formation as a measure of entanglement of a state here. $\epsilon' \in \mathscr{C}'$ is the masked state corresponding to the mixed state formed by a convex mixture of pure orthogonal states $\ket{0}$ and $\ket{1}$. The entanglement of masked state $\epsilon'$ depends on the masker $\mathscr{N}_3$. 
We plot the entanglement of formation of masked state $\epsilon'$ minimized over all the maskers satisfying Eqs.~\eqref{numerical-commuting-masking-eqn} and~\eqref{epsilon} against the mixing parameter $s$.
The nature of entanglement of formation of state $\epsilon'$ with $s$ implies that any masker could not encode the masking information contained in a set of two mixed commuting states formed by the pure states $\ket{0}$ and $\ket{1}$ into a correlation other than entanglement. Thus, it supports our analytical results, Lemma~\ref{pure-orthogonal-mask-lemma} and Theorem~\ref{neccessary-sufficient-commuting}. The vertical axis is in ebits, while the horizontal axis is dimensionless.}
\label{fig:3}
\end{figure}

We begin with the general form of any two-qubit unitary operator. Any general two-qubit unitary operator~\cite{u1} can be written in the following form:
\begin{equation*}
U = U_{A}\otimes U_{B} U_{d} V_{A}\otimes V_{B},
\end{equation*}
where $U_{d}$ is the entanglement generating gate and it is given by
\begin{equation*}
   U_{d} \coloneqq \exp{(-i\alpha_{x}\sigma_{x}\otimes\sigma_{x}-i\alpha_{y}\sigma_{y}\otimes\sigma_{y}-i\alpha_{z}\sigma_{z}\otimes\sigma_{z},
    )},
\end{equation*}
where $\alpha_x$, $\alpha_y$, and $\alpha_z$ are real numbers with $\alpha_{x,y,z}\in{[0,\frac{\pi}{2}]}$ and $\lbrace \sigma_{j} \rbrace_j$ for $j \in \lbrace x,y,z \rbrace$ represent the Pauli spin matrices.
Here, $V_{A}$, $V_{B}$, $U_{A}$, and $U_{B}$ are the local unitary operations that cannot destroy or create entanglement. Since we are interested in measuring entanglement of states
throughout our numerical calculation, so considering the operation $V_{A}\otimes V_{B}$ is enough as a local operation in the form of $U$.
Hence, $U$ can be written as
\begin{align}
    U = U_{d}V_{A}\otimes V_{B}, \label{eqn-unitary-form}
\end{align}
and we will use the above form of unitary to find the masker of the mixed states constructed by pure states, $\ket{0}$ and $\ket{1}$, in the following calculation.
The local unitaries, $V_{A}$ and $V_{B}$, can be written in general as follows:

\begin{equation*}
    V_{A}=\begin{bmatrix}
    \cos{\frac{\theta_{1}}{2}} e^{\frac{i(\psi_{1} +\phi_{1})}{2}}&  \sin{\frac{\theta_{1}}{2}}e^{\frac{-i(\psi_{1} -\phi_{1})}{2}} &\\
     -\sin{\frac{\theta_{1}}{2}}e^{\frac{i(\psi_{1} -\phi_{1})}{2}} & \cos{\frac{\theta_{1}}{2}}e^{\frac{-i(\psi_{1} +\phi_{1})}{2}}    
  \end{bmatrix},
  \end{equation*}
  \begin{equation*}
    V_{B}=\begin{bmatrix}
    \cos{\frac{\theta_{2}}{2}} e^{\frac{i(\psi_{2} +\phi_{2})}{2}}&  \sin{\frac{\theta_{2}}{2}}e^{\frac{-i(\psi_{2} -\phi_{2})}{2}} &\\
     -\sin{\frac{\theta_{1}}{2}}e^{\frac{i(\psi_{2} -\phi_{2})}{2}} & \cos{\frac{\theta_{2}}{2}}e^{\frac{-i(\psi_{2} +\phi_{2})}{2}} 
     \end{bmatrix},
\end{equation*}
where $ \theta_{1,2} \in {[0,\pi]}$,$\phi_{1,2}\in{[0,2\pi]}$, and $\psi_{1,2}\in{[0,4\pi)}$.

Let us consider that
the state in the auxiliary system is $\ket{0}$. 
The convex mixtures of the two states, $\epsilon_1$ and $\epsilon_2$, can be written as
\begin{equation*}
    \epsilon \coloneqq s\epsilon_{1}+(1-s)\epsilon_{2},
\end{equation*}
with $0 \leq s \leq 1$. 
Suppose that all the mixed states situated in the line joining two pure states, $\epsilon_1$ and $\epsilon_2$, form a set $\mathscr{C}$, which is parameterized by $s$, and the masker, let's say $\mathscr{N}_3$, masks the set of two single-qubit pure states, $\epsilon_1$ and $\epsilon_2$. Hence, we will have
\begin{align}
\label{numerical-commuting-masking-eqn}
   \Tr_{A/B}[\mathscr{N}_3 (\epsilon_{1} \otimes \ket{0}\bra{0}) \mathscr{N}_3^\dagger] &= \Tr_{A/B} [ \mathscr{N}_3 (\epsilon_{2} \otimes \ket{0}\bra{0}) \mathscr{N}_3^\dagger].
\end{align}

From Lemma \ref{masking_mixed_from_pure}, we can conclude that the masker $\mathscr{N}_3$ will also mask the quantum information encoded in the set $\mathscr{C}$. 
The masker $\mathscr{N}_3$ maps the states of composite system from $\lbrace (\epsilon \otimes \ket{0}\bra{0}) \rbrace$ to $\lbrace \epsilon' \rbrace$, i.e.,
\begin{align}
\label{epsilon}
\epsilon' &\coloneqq \mathscr{N}_3 (\epsilon \otimes \ket{0}\bra{0}) \mathscr{N}_3^\dagger. 
\end{align}
We define the new set formed by the states $\lbrace \epsilon' \rbrace$ as $\mathscr{C}'$. 
\\
We use the entanglement of formation as an entanglement measure~\cite{EF} for any two-qubit state. Entanglement of formation, $E_{F}(\rho)$ for a state $\rho$, is defined as 
\begin{equation}
E_{F}(\rho) \coloneqq h(\frac{1+\sqrt{1-(\mathbb{C}(\rho))^2}}{2}),
\nonumber
\end{equation}
\begin{align}
\textnormal{where,} \hspace{3 mm} h(x) &\coloneqq -x\log_2(x)-(1-x)\log_2(1-x)\nonumber,\\
 \mathbb{C}(\rho)) &\coloneqq \max\{0,\lambda_{1}-\lambda_{2}-\lambda_{3}-\lambda_{4}\}. \nonumber
\end{align}
Here, $\lambda_{i}$s are the eigenvalues of the Hermitian matrix $R \coloneqq \sqrt{\sqrt{\rho}\Tilde{\rho}\sqrt{\rho}}$ in decreasing order with 
$\Tilde{\rho} \coloneqq \sigma_{y}\otimes\sigma_{y}\rho^{*}\sigma_{y}\otimes\sigma_{y}$.

We will express the masker $\mathscr{N}_3$ as a two-qubit generic unitary operator $U$ (described in Eq.~\eqref{eqn-unitary-form}).
Now, we will try to find the states belonging to the set $\mathscr{C}'$ which are separable by minimizing over $\mathscr{N}_3$, keeping in mind that the masker $\mathscr{N}_3$ satisfies Eq. \eqref{numerical-commuting-masking-eqn}. Thus, we will get the separable masked states corresponding to the set, $\mathscr{C}'$.
To do so, we will minimize the entanglement of formation of the state $\epsilon'$ over all the maskers $\mathscr{N}_3$ with the condition Eq.~\eqref{numerical-commuting-masking-eqn} being satisfied.
$E_F^\textnormal{min}(\epsilon')$ denotes the entanglement of formation of the state $\epsilon'$ minimized over all the maskers.
We plot $E_F^{\textnormal{min}} (\epsilon')$ for every value of $s$ in Fig.~\ref{fig:3}. 
It is clear from Fig.~\ref{fig:3} that we only get separable masked states at $s=0.5$, and the entanglement of formation of the masked states  corresponding to this masker at the two extrema of $s$ is the maximum, i.e., $1$. 
Since to define quantum masking, we require a set consisting of at least two states, we can conclude from our numerical calculation that entanglement is necessary for masking the set of mixed states formed from convex mixtures of two single-qubit pure orthogonal states, $\ket{0}$ and $\ket{1}$, considering $\ket{0}$ as an auxiliary state.
Thus, our numerical findings concur with our analytical results in Subsection~\ref{subsection-A}.

\subsection{Numerical verification of Theorem~\ref{theorem-non-commuting}}
\label{Ex2}

\begin{figure}
     \centering
\includegraphics[width=\linewidth]{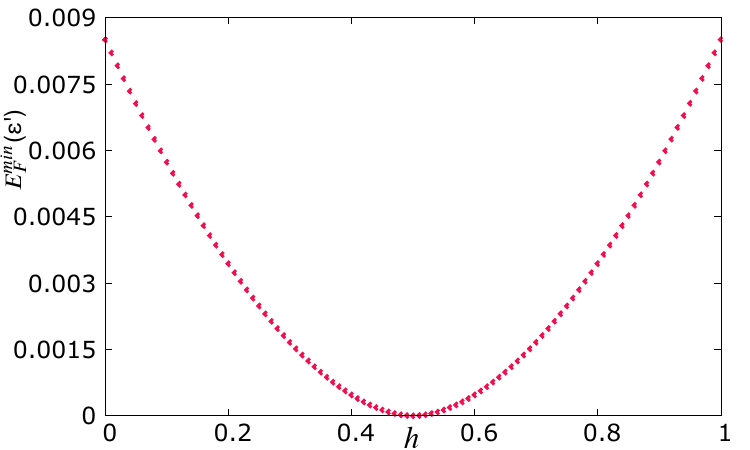}
\caption{\textit{\textbf{Nature of entanglement of formation of masked state, $\varepsilon'$, corresponding to the optimized masker that minimizes the entanglement of the set $\{\varepsilon'\}$, with respect to the mixing parameter $h$.}} 
All considerations are the same as in Fig.~\ref{fig:3}, except for the fact that in this context, the set we want to be masked is the two arbitrary single-qubit mixed non-commuting states formed by the convex mixtures of single-qubit pure non-orthogonal states, $\varrho_1$ and $\varrho_2$, parameterized by $h$.}
     \label{fig:4}
 \end{figure}

Here, we will give the numerical realization of Theorem~\ref{theorem-non-commuting}. Let us consider two single-qubit pure non-orthogonal states 
$\varrho_{1} \coloneqq \ket{0}\bra{0}$ and $\varrho_{2} \coloneqq \ket{\beta}\bra{\beta}$, with $\ket{\beta} \coloneqq \cos{\frac{\omega}{2}}\ket{0} +\sin{\frac{\omega}{2}}\ket{1}$ and $0 < \omega < \pi$. The mixed states formed from the convex combination of states, $\varrho_1$ and $\varrho_2$, can be expressed as
\begin{equation*}
    \varepsilon \coloneqq h \varrho_{1}+(1-h)\varrho_{2},
\end{equation*}
with $0\leq h \leq 1$. Now, let us construct a set of mixed states, $\lbrace \varepsilon \rbrace$, parameterized by different $h$, which will be non-commuting. Upon the action of the masker, let's say $\mathscr{N}_4$, and taking $\ket{0}\bra{0}$ as the initial state of the auxiliary system, suppose that the states $\lbrace \varepsilon \rbrace$ is masked into the state $\lbrace \varepsilon' \rbrace$, i.e.,
\begin{align*}
\label{epsilon}
\varepsilon' &\coloneqq \mathscr{N}_4 (\varepsilon \otimes \ket{0}\bra{0}) \mathscr{N}_4^\dagger. 
\end{align*}
We try to find an optimal unitary masker that not only masks a set of mixed states, $\lbrace \varepsilon \rbrace$, but also minimizes the entanglement content of the masked states corresponding to the set. 
We use the entanglement of formation to measure the entanglement of any quantum state.
As a result, we find for a specific $\omega$ that the minimum entanglement of formation of masked states corresponding to the set of states, $\lbrace \varepsilon \rbrace$, are non-zero for every $h$ except at $h = 0.5$, optimizing over all unitary maskers, as depicted in Fig.~\ref{fig:4}.

Hence, with a similar reasoning as before in Appendix~\ref{Ex1}, no masker of set corresponding to two single-qubit mixed non-commuting states formed by convex mixtures of two single-qubit pure non-orthogonal states, $\varrho_{1}$ and $\varrho_2$, will exist that could mask the quantum information contained in the set into a quantum correlation other than entanglement. Thus, our numerical example agrees well with our analytical findings in Theorem~\ref{theorem-non-commuting}. Therefore, entanglement is a necessity for masking the set of two single-qubit mixed non-commuting states.


\end{appendix}

\end{document}